\documentclass{article}

\usepackage{amsmath}
\usepackage{algorithm}
\usepackage{algcompatible}
\usepackage{graphicx,psfrag,amsmath,amssymb,epsf}
\usepackage{epsfig,graphics}

\newtheorem{theorem}{Theorem}[section]
\newtheorem{lemma}{Lemma}[section]

\newcommand{\qed}{\hfill\hbox{\rlap{$\sqcap$}$\sqcup$}}
\newenvironment{proof}{\noindent \emph{Proof.\,}}{\qed}

\algnewcommand\algorithmicreturn{\textbf{return}}
\algnewcommand\RETURN{\algorithmicreturn}
\algnewcommand\algorithmicprocedure{\textbf{procedure}}
\algnewcommand\PROCEDURE{\item[\algorithmicprocedure]}%
\algnewcommand\algorithmicendprocedure{\textbf{end procedure}}
\algnewcommand\ENDPROCEDURE{\item[\algorithmicendprocedure]}%
\algnewcommand{\algvar}[1]{{\text{\ttfamily\detokenize{#1}}}}
\algnewcommand{\algarg}[1]{{\text{\ttfamily\itshape\detokenize{#1}}}}
\algnewcommand{\algproc}[1]{{\text{\ttfamily\detokenize{#1}}}}
\algnewcommand{\algassign}{\leftarrow}

\title{The Two-Squirrel Problem and Its Relatives} 

\author{Sergey Bereg\footnote{Department of Computer Science, University of Texas at Dallas, Richardson, TX 75080, USA. Email: {\tt besp@utdallas.edu}.}
\and
Yuya Higashikawa\footnote{Graduate School of Information Science, University of Hyogo, Kobe, Japan. Email: {\tt higashikawa@gsis.u-hyogo.ac.jp}.} 
\and
Naoki Katoh \footnote{Graduate School of Information Science, University of Hyogo, Kobe, Japan. Email: {\tt naoki.katoh@gsis.u-hyogo.ac.jp}.}
\and
Manuel Lafond \footnote{Department of Computer Science, Université de Sherbrooke, Sherbrooke, Quebec J1K 2R1, Canada. Email: {\tt manuel.lafond@usherbrooke.ca}.}
\and
Yuki Tokuni\footnote{Graduate School of Information Science, University of Hyogo, Kobe, Japan. Email:  {\tt ad21o040@gsis.u-hyogo.ac.jp}.}
\and
Binhai Zhu\footnote{Gianforte School of Computing, Montana State University, Bozeman, MT 59717, USA. Email: {\tt bhz@montana.edu}.}
}

\date{}

\begin{document}

\maketitle

\begin{abstract}
In this paper, we start with a variation of the star cover problem called the Two-Squirrel problem.
Given a set $P$ of $2n$ points in the plane, and two sites $c_1$ and $c_2$,
compute two $n$-stars $S_1$ and $S_2$ centered at $c_1$ and $c_2$ respectively
such that the maximum weight of $S_1$ and $S_2$ is minimized.
This problem is strongly NP-hard by a reduction from Equal-size Set-Partition with Rationals. 
Then we consider two variations of
the Two-Squirrel problem, namely the Two-MST and Two-TSP problem, which are both NP-hard. The NP-hardness for the latter is obvious while the former needs a non-trivial reduction from Equal-size Set-Partition with Rationals.
In terms of approximation algorithms, 
for Two-MST and Two-TSP we give factor
3.6402 and $4+\varepsilon$ approximations respectively. 
Finally, we also show some interesting polynomial-time solvable cases for Two-MST.
\end{abstract}

\section{Introduction}

Imagine that two squirrels try to fetch and divide $2n$ nuts to their nests. Since each time a squirrel
can only carry a nut back, this naturally gives the following problem: they should travel along the edges of an $n$-star, centered at the corresponding nest, such that each leaf (e.g., nut) is visited exactly once (in and out) and the maximum distance they visit should be minimized (assuming that they travel at the same speed, there is no better way to enforce the fair division under such a circumstance).
See Figure 1 for an illustration.
\begin{figure}[htbp]
    \centering
    \includegraphics[width=0.70\textwidth]{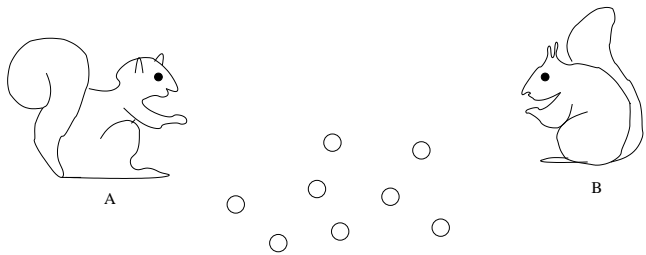}
    \caption{Two squirrels $A$ and $B$ try to fetch and divide $2n$ nuts.}
    \label{fig1}
\end{figure}

A star $S$ is a tree where all vertices are leaves except one (which is called the {\em center} of the star). An $n$-star is a star with $n$ leaf nodes. When the edges in $S$ carry weights, the weight of $S$ is the sum of weights of all the edges in $S$. Given two points $p,q$ in the plane, with $p=(x_p,y_p)$ and $q=(x_q,y_q)$, we define the Euclidean distance between $p,q$ as $d(p,q)=|pq|=\sqrt{(x_p-x_q)^2+(y_p-y_q)^2}$ and the $L_1$ or Manhattan distance between them is defined as $d_1(p,q)=|x_p-x_q|+|y_p-y_q|$. 

Formally, the {\em Two-Squirrel} problem can be defined as: Given a set $P$ of $2n$ points in the plane and two extra point sites $c_1$ and $c_2$, compute two $n$-stars $S_1$ and $S_2$ centered at $c_1$ and $c_2$ respectively such that each point $p_j\in P$ is a leaf in exactly one of $S_1$ and $S_2$; moreover, the maximum weight of $S_1$ and $S_2$ is minimized. Here the weight of an edge $(c_i,p_j)$ in $S_i$ is $w(c_i,p_j)=d(c_i,p_j)$ for $i=1,2$. 
One can certainly consider a variation of the two-squirrel problem where the points are given as pairs $(p_{2i-1},p_{2i})$ for $i=1,...,n$, and the problem is to split all the pairs (i.e., one to $c_1$ and the other to $c_2$) such that maximum weight of the two resulting stars is minimized. We call this version {\em Dichotomy Two-Squirrel}. 

A more general (and probably more interesting) version of the problem is when the two squirrels only need to split the $2n$ nuts and each could travel along a Minimum Spanning Tree (MST) of the $n$ points representing the locations of the corresponding nuts, which we call the {\em Two-MST} problem: Compute a partition of $P$ into $n$ points each, $P_1$ and $P_2$, such that the maximum weight of the MST of $P_1\cup\{c_1\}$ and
$P_2\cup\{c_2\}$, i.e., $\max\{w(P_1\cup\{c_1\}),w(P_2\cup\{c_2\})\}$, is minimized.
Similarly, we could replace MST with TSP to have the {\em Two-TSP} problem.

Covering a (weighted) graph with stars or trees (to minimize the maximum weight of them) is a well-known NP-hard problem in combinatorial optimization \cite{DBLP:conf/random/EvenGKRS03}, for which constant factor approximation is known. Recently, bi-criteria approximations are also reported \cite{DBLP:conf/approx/GamlathG20}. 
In the past, a more restricted version was also investigated on graphs \cite{DBLP:conf/tamc/ZhaoZ07}. 
Our Two-Squirrel problem can be considered a special geometric star cover problem where the two stars are disjoint though are of the same cardinality, and the objective function is also to minimize the maximum weight of them.

It turns out that, when the coordinates of points are rationals, both Two-Squirrel and Dichotomy Two-Squirrel are strongly NP-hard (under both the Euclidean and $L_1$ metric, though we focus only on the Euclidean case in this paper). The proofs can be directly from two variations of the famous Set-Partition problem \cite{DBLP:books/fm/GareyJ79,karp1972}, namely, Equal-Size Set-Partition with Rationals and Dichotomy Set-Partition with Rationals, which are both strongly NP-hard with the recent result
by Wojtczak \cite{DBLP:conf/csr/Wojtczak18}.
We then show that Equal-size Set-Partition with Rationals can be reduced to Two-MST in polynomial time, which indicates that Two-MST is NP-hard. (Note
that in this proof, the constructed points have real coordinates.) On the
other hand, Two-TSP is obviously NP-hard as the TSP problem is NP-hard.

For the approximation algorithms, both Two-Squirrel and Dichotomy Two-Squirrel admit a FPTAS (note that this does not contradict the known result that a strongly NP-hard problem with an integral objective function cannot be approximated with a FPTAS unless P=NP, simply because our objective functions are not integral). This can be done by first designing
a polynomial-time dynamic programming algorithm through
scaling and rounding the distances to integers, obtaining the corresponding optimal solutions, and
then tracing back to obtain the approximate solutions. The approximation algorithm for Two-MST is more tricky; in fact, with a known lower bound by Chung and Graham related to the famous Steiner Ratio Conjecture \cite{CG85}, we show that a factor 3.6402
approximation can be obtained. Using a similar method, we show that Two-TSP can be approximated with a factor of $4+\varepsilon$.

In the end, we show two interesting polynomial-time solvable cases: when all the points in $P$ and the two sites are on the X- and Y-axis, the problems are polynomially solvable under both the $L_1$ and $L_2$ distances. The running times are $O(n^4)$ and $O(n^{13})$ respectively.

The paper is organized as follows. In Section 2, we give some necessary definitions.
In Section 3, we present our NP-hardness result for the Two-MST problem. In Section 4 we present the approximation algorithms for Two-MST and Two-TSP. In Section 5, we show the special polynomial-time solvable cases. And in Section 6 we conclude the paper.

\section{Preliminaries}

In this section, we first define
Equal-size Set-Partition for Rationals and Dichotomy Set-Partition for Rationals which are generalizations of Set-Partition ~\cite{DBLP:books/fm/GareyJ79,karp1972}.

In Dichotomy Set-Partition with Rationals, we are given a set $E$ of $2n$ positive rationals numbers (rationals, for short)
with $E=E'_1\cup E'_2\cup\cdots E'_n$ such that $E'_i=\{a_{i,1},a_{i,2}\}$ is a 2-set (or, $E'_i=(a_{i,1},a_{i,2})$, i.e., as a pair) and
the problem is to decide whether $E$ can be partitioned into $E_1$ and $E_2$
such that every two elements in $E'_i$ is partitioned into $E_1$ and $E_2$
(i.e., one in $E_1$ and the other in $E_2$ --- clearly $|E_1|=|E_2|=n$) and $\sum_{a\in E_1}a=\sum_{b\in E_2}b$. (Equal-size Set-Partition with Rationals is simply a special case of Dichotomy Set-Partition with Rationals where $E$ is given as a set of $2n$ rationals, i.e., $E=\{a_1,a_2,\cdots,a_{2n}\}$ and $E'_i$'s are not given.)

With integer inputs, both Dichotomy Set-Partition and Equal-size Set-Partition, like their predecessor Set-Partition, can be shown to be
weakly NP-complete. Recently, Wojtczak proved that even with rational inputs, Set-Partition is strongly NP-complete \cite{DBLP:conf/csr/Wojtczak18}.
In fact, the proof by Wojtczak implied that Dichotomy Set-Partition and Equal-size Set-Partition are both strongly NP-complete --- because in this reduction from a special 3-SAT each pair $x_i$ and $\bar{x}_i$ are associated with two unique rational numbers which must be split in two parts. So we re-state this theorem by Wojtczak.

\begin{theorem}
Equal-size Set-Partition with Rationals and Dichotomy Set-Partition with Rationals are both strongly NP-complete.
\end{theorem}

It is straightforward to reduce Equal-size Set-Partition with Rationals
to Two-Squirrel (with rational coordinates) and Dichotomy Set-Partition with Rationals to Dichotomy Two-Squirrel (with rational coordinates), as each point is directly connected to either $c_1$ or $c_2$. Hence, both
Two-Squirrel and Dichotomy Two-Squirrel are strongly NP-hard when the
coordinates of the input points are rational.

Coming to Two-MST, the story is quite different. Since the structure of
an MST is not fixed (i.e., even if we know that two points $u,v\in P$ belong to
$T_1$, the MST of $P_1\cup\{c_1\}$, we do not know how $u,v$ are connected before $T_1$ is actually computed). Nonetheless, we show in the next section that Two-MST is NP-hard.

\section{NP-hardness for Two-MST}

In this section, we prove that the Two-MST problem (2-MST for short), is NP-hard. (Our construction requires that the coordinates of the points are real numbers.)
Recall that in the 2-MST problem, one is given a set $P$ of $2n$ points in the plane, together with two point sites $c_1$ and $c_2$, the objective is to compute two MST $T_1$ and $T_2$ each containing $n$ points in $P$ (and $c_1$ and $c_2$ respectively) such that the maximum weight of $T_1$ and $T_2$, $\max\{w(T_1),w(T_2\}$, is minimized. Here the weight of any edge $(p_i,p_j)$ or $(p_i,c_k)$ in $T_k,k=1..2$, is the Euclidean distance
between the two corresponding nodes.
We reduce Equal-size Set-Partition for Rationals \cite{DBLP:conf/csr/Wojtczak18} to 2-MST in the following. Note that in the proof by Wojtczak \cite{DBLP:conf/csr/Wojtczak18}, a set $S$ of $2n$ rationals, with a total sum of $2n$, were constructed such that the only partition is to partition them into two equal-size sets with $n$ rationals, each having a sum of value $n$. 

\begin{theorem}
Two-MST is NP-hard.
\end{theorem}

\begin{proof}
We reduce Equal-size Set-Partition with Rationals to Two-MST.
Note that, given $E=\{a_1,a_2,\cdots,a_{2n}\}$ where each $a_i~(i=1..2n)$ is a rational number and $\sum_ia_i=2t$, for Set-Partition with Rationals we need to
partition $E$ into two sets $E_1$ and $E_2$ such that $|E_1|=|E_2|$ and the rationals
in $E_1$ and $E_2$ sum the same, i.e., $t=\sum_{a\in E_1}a=\sum_{b\in E_2}b$. We construct $10n+4$ points in $P$ as well as 2 point sites $c_1$ and $c_2$. We first show our ideas, then follow with the construction of these
points with coordinates --- mostly along the X-axis.

The building block of each $a_i$ is a rectangle $B_i=(b_{i,1},b_{i,2}, b_{i,4},b_{i,3})$ in clockwise order with $b_{i,1}$ being the top-left corner point; in addition, $p_{i}$ (on the X-axis) is the center of this rectangle $B_i$ (see Fig.2 (II)). In other words, each $a_i$ will be transformed into a group of 5 points. The horizontal edge length of 
$B_i$ is $24a_i$ and the height of $B_i$ is $10a_i$; hence the distance from the center $p_i$ to any of the corner point is $13a_i$. The crucial point is that, 
at $B_i$, if $T_1$ and $T_2$ start at $b_{i,1}$ and $b_{i,3}$ respectively, then one of them would include $p_i$ and ending at $b_{i,2}$ and $b_{i,4}$ respectively (or vice versa). As a matter of fact, the difference of the parts of $T_1$ and $T_2$ spanning $B_i\cup \{p_i\}$ is $2\times 13a_i-2\times 12a_i=2a_i$.
We place the $B_i$'s in a way such that the right edge of $B_i$ and the left edge of $B_{i+1}$ form an isosceles trapezoid $T_i$, symmetric along the X-axis, such that the non-vertical edges have a length of $2t$ (note that $2t>a_i, 2t>a_{i+1}$). As a matter of fact, going from left to right, if $T_1$ (resp. $T_2$) includes $b_{i,2}$ (resp. $b_{i,4}$), then the shortest paths from them to reach $B_{i+1}$ are $<b_{i,2},b_{i+1,1}>$ and $<b_{i,4},b_{i+1,3}>$ respectively, which both have a length of $2t$.

\begin{figure}[htbp]
\psfrag{T0}{$T_{0}$}
\psfrag{T1}{$T_{1}$}
\psfrag{T2}{$T_{2}$}
\psfrag{p1}{$p_{1}$}
\psfrag{B1}{$B_{1}$}
\psfrag{p2}{$p_{2}$}
\psfrag{B2}{$B_{2}$}
\psfrag{Bi}{$B_{i}$}
\psfrag{bi1}{$b_{i,1}$}
\psfrag{bi2}{$b_{i,2}$}
\psfrag{bi3}{$b_{i,3}$}
\psfrag{bi4}{$b_{i,4}$}
\psfrag{pi}{$p_{i}$}
\psfrag{10ai}{$10a_{i}$}
\psfrag{24ai}{$24a_{i}$}
\psfrag{c1}{$c_1$}
\psfrag{c2}{$c_2$}
\psfrag{2t}{$2t$}
    \centering
    \includegraphics[width=0.65\textwidth]{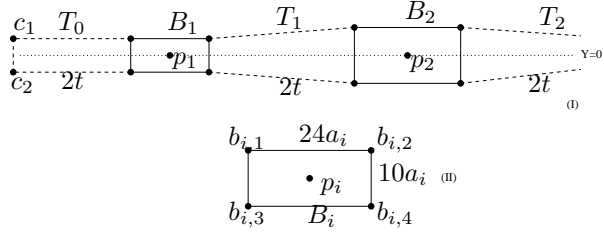}
    \caption{Illustration for the reduction from Equal-size Set-Partition with Rationals to 2-MST, the left part (I). Block $B_i$, note that the distance from the center $p_i$ to any of the 4 corners is $13a_i$ (II).}
   \label{fig3}
\end{figure}

At the end of $B_{2n}$, we construct four points $b_{2n+1}=b_{2n+2}$ (on the X-axis), $q$ and $r$. They form a regular triangle with $d(b_{2n+1},q)=d(b_{2n+1},r)=d(q,r)=4nt$. As the distance $d(q,r)$ is so large (compared with the optimal solution for 2-MST), the optimal solution
must split them in a way such that $\{b_{2n+1},q\}\in T_1$ and $\{b_{2n+2},r\}\in T_2$ or vice versa. Moreover, we can set $d(b_{2n,3},b_{2n+1})=d(b_{2n,4},b_{2n+1})=4nt$; i.e., $<b_{2n,3},b_{2n+1},b_{2n,4}>$ form an isoceles triangle with long edge
length $4nt$ (or we can say $<b_{2n,3},b_{2n+1},b_{2n+2},b_{2n,4}>$ form a degenerate isoceles trapezoid $T_{2n}$ with edge length $4nt$.
Obviously, in the optimal solution $b_{2n,3}$ and $b_{2n,4}$ must be split into $T_1$ and $T_2$ respectively, or vice versa.

We briefly discuss the coordinates of the points constructed; in fact, they could be constructed in an incremental way. First set $c_1=(0,10a_1), c_2=(0,-10a_1)$, and construct the group of 5 points as the vertices and center of $B_1$, with $b_{1,1}=(2t,10a_1)$, $b_{1,2}=(2t+24a_1,10a_1)$, $b_{1,3}=(2t,-10a_1)$, $b_{1,4}=(2t+24a_1,-10a_1)$ and $p_1=(2t+12a_1,0)$.
Then we construct $T_i$ and $B_{i+1}\cup\{p_i\}$ for $i=1$ to $2n$ incrementally. WLOG, let $a_i\leq a_{i+1}$ and the coordinates of $b_{i,2}$ and $b_{i,4}$ be $b_{i,2}=(x_i,10a_i)$ and $b_{i,4}=(x_i,-10a_i)$ respectively. Then the coordinates of points in $B_{i+1}\cup \{p_{i+1}\}$ are
$$b_{i+1,1}=(x_i+\sqrt{(2t)^2-(10(a_{i+1}-a_i))^2},10a_{i+1}),$$,
$$b_{i+1,2}=(x_i+\sqrt{(2t)^2-(10(a_{i+1}-a_i))^2}+24a_{i+1},10a_{i+1}),$$
$$b_{i+1,3}=(x_i+\sqrt{(2t)^2-(10(a_{i+1}-a_i))^2},-10a_{i+1}),$$
$$b_{i+1,4}=(x_i+\sqrt{(2t)^2-(10(a_{i+1}-a_i))^2}+24a_{i+1},-10a_{i+1})$$ and
$$p_{i+1}=(x_i+\sqrt{(2t)^2-(10(a_{i+1}-a_i))^2}+12a_{i+1},0).$$
The coordinates for $b_{2n+1}$ and $b_{2n+2}$ are
$(x_{2n}+\sqrt{(4nt)^2-(10a_{2n})^2},0)$, and the coordinated of $q$ and $r$
are
$q=(x_{2n}+\sqrt{(4nt)^2-(10a_{2n})^2}+2\sqrt{3}nt, 2nt)$ and
$r=(x_{2n}+\sqrt{(4nt)^2-(10a_{2n})^2}+2\sqrt{3}nt, -2nt)$. 
Note that to the right of $B_1$, the points are virtually all having real coordinates. See Fig 2. (I) and Fig. 3 for the construction.

 
We show next that Equal-size Set-Partition with Rationals has a solution iff the 2-MST instance $P\cup\{c_1,c_2\}$ admits a solution with optimal weight of $(12n+2)t$. 

``If part'': If $E$ can be partitioned into $E_1$ and $E_2$ such that $\sum_{a\in E_1}a=\sum_{b\in E_2}b=t$, we show how to construct two MST's as follows. Up to $B_{2n}$, we include all the points above the X-axis to $T_1$ and all the points below the X-axis to $T_2$. For $p_i$'s, if $a_i\in E_1$ then we include $p_i$ in $T_1$, if $a_i\in E_2$ then we include $p_i$ in $T_2$ (each will incur a cost of $2a_i$).
We then include $b_{2n+1}$ and $b_{2n+2}$ (and $q$ and $r$) to $T_1$ and $T_2$ respectively.
Obviously we have $|T_1|=|T_2|=5n+3$, and the weight of them are both
$(12n+2)t$.

``Only-if part'': Now suppose that points in $P$ are partitioned
into $P_1$ and $P_2$ such that the MST's of $P_1\cup\{c_1\}$ and $P_2\cup\{c_2\}$ are $T_1$ and $T_2$ respectively, and the maximum weight of $T_1$ and $T_2$ is $(12n+2)t$. Following the previous argument, we must split $q$ and $r$ (hence also $b_{2n+1}$ and $b_{2n+2}$, and subsequently
$b_{2n,2}$ and $b_{2n,4}$) into $T_1$ and $T_2$ to have a weight less than $16nt$. Similarly, we need to split $b_{1,1}$ and $b_{1,3}$ into $T_1$ and $T_2$ as otherwise we would have a solution larger than $(12n+2)t$ --- since $d(c_1,b_{1,3})>d(c_1,b_{1,1})=2t$ and $d(c_2,b_{1,1})>d(c_2,b_{1,3})=2t$.
Likewise, not splitting $b_{1,1}$ and $b_{1,3}$ into $T_1$ and $T_2$, e.g.,
including both of them in $T_1$ or $T_2$, would incur a cost of $2t+10a_1>2t$, which would lead to a higher total cost.

We now show with induction that the current optimal solution (say $T_1$) for points up to $B_i$ is $2it$ (the major cost) plus the cost of including some center $p_j$'s ($1\leq j\leq i$); moreover, $T_1$ must include $b_{i,2}$, $T_2$ must include $b_{i,4}$ and the cost of the other MST $T_2$ is minimized. The basis is obvious: since $T_1$ must include $b_{1,1}$ and $T_2$ must include $b_{1,3}$, to reach the end of $B_1$ (i.e., $b_{1,2}$ and $b_{1,4}$), $T_1$ needs to include $b_{1,2}$ and $p_1$ to maintain the optimality of a local solution ($2t+26a_1$), and $T_2$ must include $b_{1,4}$ to have a cost of $2t+24a_1$. Note that if we let $T_1$ include $p_1$ and $b_{1,4}$, and $T_2$ include $b_{1,2}$, although the cost of $T_1$ remains the same ($2t+26a_1$), the cost of $T_2$ becomes $2t+26a_1$, which is not minimized anymore.

Now assuming the inductive hypothesis holds for $i$, let us consider $B_{i+1}$. In very much the same way, let the local optimal solution (say $T_1$) end at $b_{i,2}$, and $T_2$ end at $b_{i,4}$, with both the major cost being $2it$. Clearly, in covering points in $B_{i+1}$, $T_1$ (resp, $T_2$) should not include $b_{i+1,3}$ (resp. $b_{i+1,1}$) as that will increase the major cost to more than $2(i+1)t$ (since in $T_i$, $d(b_{i,2},b_{i+1,3})>d(b_{i,2},b_{i+1,1})=2t$ and
$d(b_{i,4},b_{i+1,1})>d(b_{i,4},b_{i+1,3})=2t$). Then, for the same argument as in the basis, if $T_1$ includes $b_{i+1,4}$ and $T_2$ includes $b_{i+1,2}$ then the cost of $T_2$ is not minimized.

At this point, it can be seen that the optimal solution boils down to
split $p_i$'s to $T_1$ and $T_2$. As we have $2n$ $p_i$'s and the
splitting of each $p_i$ would incur a cost of $2a_i$, by symmetry,
the optimal solution must split them into $T_1$ and $T_2$ such that each
would incur an additional cost of $2t$ (note that $\sum_{1\leq i\leq 2n}a_i=2t$), for a total cost of $2(2n)t+2t+(4nt+4nt)=(12n+2)t$. The splitting of these $a_i$'s in $T_1$ and $T_2$ would 
return us a solution for Equal-size Set-Partition with Rationals, i.e.,
if $a_i$ is in $T_1$ then $E_1\leftarrow E_1\cup \{a_i\}$, and
if $a_i$ is in $T_2$ then $E_2\leftarrow E_2\cup \{a_i\}$; moreover
$\sum_{a\in E_1}a=\sum_{b\in E_2}b=t$.

This reduction obviously takes linear time, hence the theorem is proven.
\end{proof}

We comment that with this proof, a variation of 2-MST, e.g., even if $c_1$ and $c_2$ are not given in advance, remains NP-hard. Also, with a minor modification we could show that Two-MST is NP-hard under the $L_1$ distance as well. In addition, Two-TSP is obviously NP-hard: given a set of points $P$ and suppose we want to compute a TSP of $P$. We just create another copy of $P$, $P'$ and translate $P'$ to be far away from $P$ (say, by a distance of 10 times the diameter of $P$), then fix a point $p$ in $P$ as $c_1$ and the corresponding copy $p'$ in $P'$ as $c_2$. Then the optimal solution for TSP for $P$ is exactly the same
as the Two-TSP solution for $P\cup P'\cup\{c_1,c_2\}$. 

In the next section, we present constant-factor approximations for Two-MST and Two-TSP.

\begin{figure}[htbp]
\psfrag{T2n-1}{$T_{2n-1}$}
\psfrag{T2n}{$T_{2n}$}
\psfrag{p2n}{$p_{2n}$}
\psfrag{B2n}{$B_{2n}$}
\psfrag{10a2n}{$10a_{2n}$}
\psfrag{24ai}{$24a_{i}$}
\psfrag{a2n+1}{$b_{2n+1}$}
\psfrag{b2n+1}{$b_{2n+2}$}
\psfrag{q}{$q$}
\psfrag{r}{$r$}
\psfrag{2t}{$2t$}
\psfrag{4nt}{$4nt$}
    \centering
    \includegraphics[width=0.65\textwidth]{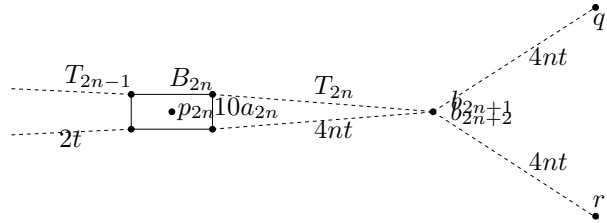}
    \caption{Illustration for the reduction from Equal-size Set-Partition with Rationals to 2-MST, the right part.}
   \label{fig3}
\end{figure}

\section{Constant-factor Approximations for Two-MST and Two-TSP}

Note that, when the coordinates of points are rational, both Two-Squirrel and Dichotomy Two-Squirrel admit a FPTAS. This can be done, as suggested by Wojtczak \cite{DBLP:conf/csr/Wojtczak18} for the corresponding counterparts of Set-Partition (with rationals), by first designing
a polynomial-time dynamic programming algorithm through
scaling and rounding the distances to integers, obtaining the corresponding optimal solutions, and then tracing back to obtain the approximate solutions. This method does not work for 2-MST and 2-TSP. In fact, in this section, we have to design constant-factor approximations for them separetely.

\subsection{A 3.6402-Approximation for 2-MST}

Recall that we are given a set $P$ of $2n$ points in the plane, and two sites $c_1$ and $c_2$,
partition $P$ into two sets $P_1$ and $P_2$ each of size $n$ such that the maximum weight of MST's for $P\cup\{c_1\}$ and $P_2\cup \{c_2\}$ is minimized.

\vspace{2mm}
{\bf Algorithm for 2-MST}.
\begin{enumerate}
\item
Compute $T$, a MST of $P\cup\{c_1,c_2\}$, using Kruskal's algorithm. 
Let $e$ be the last edge added to $T$ joining trees $T_1$ and $T_2$.

\item
If $c_1$ and $c_2$ are in different trees $T_i$ and each tree $T_i$ contains exactly $n+1$ vertices, then $(T_1,T_2)$ is a solution. 
Otherwise do Step 3.

\item
Split $P$ into $P'_1$ and $P'_2$ each of size $n$ arbitrarily.
Compute $T'_1,T'_2$, an MST of $P'_1\cup\{c_1\},P'_2\cup\{c_2\}$, respectively.
Return $(T'_1,T'_2)$.
\end{enumerate}

\begin{theorem}
The algorithm for 2-MST has an approximation ratio 3.6402 and it runs in $O(n\log n)$ time.
\end{theorem}

\begin{proof}
Let $T$ be an MST of $P\cup \{c_1, c_2\}$.  
If the algorithm stops at Step 2, then the two trees are optimal.

Suppose that the algorithm stops at step 3. 
Let $T_1^*$ and $T_2^*$ be two optimal trees with $c_1\in T_1^*$ and $c_2\in T_2^*$.
Viewing points in $P'_2\cup\{c_2\}$ as Steiner points, by the bound of Chung and Graham \cite{CG85}, we have $w(T'_1) \leq (1/0.82416874)\cdot w(T)\leq 1.2134\cdot w(T)$.
Similarly, we also have 
$w(T'_2) \leq 1.2134\cdot w(T)$.
The approximate solution {APP} satisfies ${APP} \leq 1.2134\cdot w(T)$.

To obtain the final factor, let {OPT} be the maximum weight of $T_1^*$ and $T_2^*$.
Let $V_1^*$ and $V_2^*$ be the sets of vertices of $T_1^*$ and $T_2^*$, respectively.
Let $v_1$ and $v_2$ be two vertices of $V_1^*$ and $V_2^*$ such that the distance between $v_1$ and $v_2$ is minimized.
By taking the union of $T_1^*$ and $T_2^*$, and adding an edge between $v_1$ and $v_2$, we obtain a spanning tree.
Thus, $w(T_1^*) + w(T_2^*) + d(v_1,v_2) \geq w(T)$, since $T$ is a minimum spanning tree of $P\cup\{c_1,c_2\}$. 

Next we show that ${OPT} \geq d(v_1, v_2)$.
Suppose to the contrary that ${OPT}<d(v_1, v_2)$. 
Then the weight of every edge in $T_1^*$ and $T_2^*$ is less than $d(v_1, v_2)$.
If $d(v_1, v_2)$ is less than or equal to the weight of edge $e$ found in Step 1, then $T_1^*$ and $T_2^*$ will be found in Step 2.
Therefore $d(v_1, v_2) > w(e)$.
Since $(T_1,T_2)\ne (T_1^*,T_2^*)$, there are two points $u_1\in T_1^*$ and $u_2\in T_2^*$ such that $T_1$ or $T_2$ contains both $u_1$ and $u_2$.
Then the path between $u_1$ and $u_2$ in this tree contains an edge across the cut of $V_1^*$ and $V_2^*$. 
Since its weight is at most $w(e)$, it contradicts $d(v_1, v_2) > w(e)$.  

Thus we obtain $$
w(T_1^*) + w(T_2^*) + d(v_1,v_2) \leq {OPT} + {OPT} + {OPT} = 3\cdot{OPT}.$$
Combined with the above, this gives 
$APP \leq 1.2134\cdot w(T) \leq 1.2134\cdot (3\cdot {OPT})=3.6402\cdot{OPT}.$
\qed
\end{proof}

\subsection{A $(4+\varepsilon)$-Approximation for Two-TSP}

First let $P_1$ be the subset of points closer to $c_1$, and $P_2$ the subset of points closer to $c_2$ (ties are broken arbitrarily). For our algorithm, we first compute an MST $T$ of $P\cup\{c_1,c_2\}$, using Kruskal's algorithm. 
Let $e$ be the last edge added to $T$ joining trees $T_1$ and $T_2$.
If $c_1$ and $c_2$ are in different trees $T_i$ and each tree $T_i$ contains exactly $n+1$ vertices, then compute the approximate TSP $O_i$ for points on $T_i, i=1,2$, by doubling the edges in $T_i$, and return $(O_1,O_2)$ as a solution. (Note that $w(O_i)\leq 2\cdot {OPT}$, where {OPT} is the optimal solution value for Two-TSP. We could use a better approximation for this part, but it does not affect the final approximation factor.)

If the above condition does not hold, then let $O^*$ be an optimal TSP of $P\cup \{c_1, c_2\}$.  
Traverse $O^*$ from $c_1$ either in CW or CCW order to hit the
$n$-th point $q\in P$ on $O^*$ without hitting $c_2$ ($q$ must exist in one direction, say CCW). $O_1$ is obtained by connecting 
$c_1$ and $q$; and $O_2$ is obtained by connecting the two points neighboring $c_1$ and $q$ on $O^*$ but do not belong to $O_1$.
Clearly, we have $w(O_1)\leq w(O^*)$, as $d(c_1,q)$ is bounded above by the path from $q$ to $c_1$ (in CCW order).
Similarly, we have $w(O_1)\leq w(O^*)$.
Note that since computing $w(O^*)$ is NP-hard, we could make use of any existing PTAS \cite{DBLP:journals/jacm/Arora98,DBLP:journals/siamcomp/Mitchell99}, hence we have
$w(O_1)\leq (1+\epsilon)\cdot w(O^*)$ and
$w(O_2)\leq (1+\epsilon)\cdot w(O^*)$. Then, the approximation solution value {APP} satisfies that
$${APP}=\max\{w(O_1),w(O_2)\}\leq (1+\epsilon)w(O^*).$$

To obtain the final factor, let $O^*_1$ and $O^*_2$ be the two TSP's of the optimal solution, and let {OPT} be the maximum weight of $O^*_1$ or $O^*_2$.
By taking the union of $O^*_1$ and $O^*_2$, and doubling the edge between $c_1$ and $c_2$, we obtain a TSP for $P\cup\{c_1,c_2\}$.
Thus, $w(O^*_1) + w(O^*_2) + 2d(c_1,c_2) \geq w(O^*)$, since $O^*$ is an optimal TSP for $P\cup\{c_1,c_2\}$.

Next we show that ${OPT} \geq d(c_1, c_2)$.
If the optimal solution splits $P$ into $P_1$ and $P_2$, our algorithm just returns a factor-2 approximation of it. Now assume that the optimal solution does not do that.  
This means that $O^*_1$ has a point of $P_2$, or $O^*_2$ has a point of $P_1$.
Let $p \in O^*_1 \cap P_2$, then the two paths from $c_1$ to $p$ on $O^*_1$ shows that ${OPT} \geq 2\cdot d(c_1,p)\geq d(c_1,c_2)$. 
The same inequality holds if $p \in O^*_2 \cap P_1$.

Thus we obtain $$w(O^*)\leq w(O^*_1) + w(O^*_2) + 2d(c_1, c_2) \leq {OPT} + {OPT} + 2\cdot {OPT} = 4\cdot{OPT}.$$
Combined with the above, this gives 
$APP \leq (1+\epsilon)\cdot w(O^*) \leq (1+\epsilon)\cdot (4\cdot {OPT})=(4+\varepsilon)\cdot{OPT},$
by setting $\varepsilon=4\epsilon$.

The running time of the algorithm is dominated by the PTAS for computing the TSP of a set of $n$ points \cite{DBLP:journals/jacm/Arora98,DBLP:journals/siamcomp/Mitchell99}. Hence we have the following theorem.

\begin{theorem}

Two-TSP can be approximated with a factor-($4+\varepsilon$) approximation algorithm which runs in polynomial time (in $n$ and $1/\varepsilon$).
\end{theorem}

In the next section, we present some polynomial-time solvable cases for
Two-MST.

\section{Polynomially-solvable Cases for Two-MST}

\subsection{The 1-dimensional case: all data points are on a line}

First consider the 1-dimensional case where $P\cup\{c_1,c_2\}\subset R$, the set of real numbers.
WLOG, assume that $x(c_1)\le x(c_2)$. Let $P=\{p_1,\dots,p_{2n}\}$ be sorted by $x$-coordinates.
It can be easily shown that the optimal partition of $P$ is $P_1=\{p_1,\dots,p_n\}$ and $P_2=\{p_{n+1},\dots,p_{2n}\}$.
Hence this version can be solved in $O(n\log n)$ time with sorting. And this is optimal as we need to return the
two MST's which together give the sorted ordering of $P$.

\subsection{Points on the X- and Y-axis and under the Manhattan distance}

In this subsection, we study an interesting variation when
the distance is Manhattan ($L_1$) and all the data points (including $c_1$ and $c_2$) are on the X- and Y-axis. We call this version the X+Y case, which we show
to be solvable in polynomial time as follows. 

The following definition hold for both $L_1$ and $L_2$. The {\em maximal} segment of a tree $T_i$ on an half-axis $H=\{(x,0)~|~x\ge 0\}$ (resp. $H=\{(0,x)~|~x\ge 0\}$) is a segment between the leftmost (resp. bottom-most) vertex of $T$ in $H$ and the rightmost (resp. top-most) vertex of $T$ in $H$ (if $c_i$ is not on $H$); otherwise $H$ contains at most two maximal segments: one is from the leftmost (resp. bottom-most) vertex to the vertex before $c_i$, and the second one is from the vertex after $c_i$ to the rightmost (resp. top-most) vertex. (Similar definitions can be made for the half-axis along $-\infty$ directions.) We first prove the following lemma. 

\begin{lemma}
When all the points in $P$ and two sites $c_1$ and $c_2$ are on the X- and Y-axis, for 2-MST under the $L_1$ metric there is an optimal solution such that all the edges in the two MST's $T_1$ and $T_2$ are on the X-axis and Y-axis; moreover, the maximal segment of $T_1$ and $T_2$ on any half-axis are disjoint.
\label{lem3}
\end{lemma}

\begin{proof}
The first part of the proof goes as follows. Suppose in one of the MST's, say $T_1$, one of the edge between $(x_i,0)$ and $(0,y_j)$ is through $(x_i,y_j)$. Then by the property of $L_1$, we could connect $(x_i,0)$ to $(0,y_j)$ through the origin $o=(0,0)$.
The new $T'_1$ either has the same weight as $T_1$ (when both the segments between $(0,0)$ and $(0,x_i)$, and between $(0,0)$ and $(y_j,0)$ are not in $T_1$), or has a smaller weight as $T_1$ (when one of the segments between $(0,0)$ and $(0,x_i)$, and between $(0,0)$ and $(y_j,0)$ is already in $T_1$).

We now assume that the optimal solution of this X+Y instance for 2-MST under $L_1$ metric preserves this property that we have just proved. Note that if $T_i$ is in the optimal solution of 2-MST, all the edges of $T_i$ must be on the two axes; and if $c_i$ is on one axis, say Y-axis, then
the points of $T_i$ on the Y-axis must form at most two maximal segments, with $c_i$ in between them.

For the second part of the proof, suppose on the half-axis $(o,(+\infty,0))$ of X-axis we
have segments of points like $P'=\langle p_{1,1},\cdots,p_{1,q}$, $p_{2,1},\cdots, p_{2,r}$, $p_{1,q+1},\cdots,p_{1,q+s}$ $\rangle$, where $p_{1,i}\in T_1$ and $p_{2,j}\in T_2$; moreover, we can assume that $c_1$ and $c_2$ are out of these segments (if not, we just choose the overlapping segments not containing $c_1$ and $c_2$). Then we can obviously switch the points in the middle without increasing the weight of $T_1$ and $T_2$ as follows.
If $c_1$ and $c_2$ are both to the left of $p_{1,1}$, we just assign the leftmost $r$ points in $P'$ to $T_2$ and the remaining ones to $T_1$;
if $c_1$ and $c_2$ are both to the right of $p_{1,q+s}$, we just assign the rightmost $r$ point in $P'$ to $T_2$ and the remaining ones to $T_1$.
If $c_1$ is to the left of $p_{1,1}$ and $c_2$ is to the right of $p_{1,q+s}$, we just assign the rightmost $r$ points in $P'$ to $T_2$ and the remaining ones to $T_1$.
If $c_1$ is to the right of $p_{1,q+s}$ and $c_2$ is to the left of $p_{1,1}$, we just assign the leftmost $r$ points in $P'$ to $T_2$ and the remaining ones to $T_1$.
\end{proof}

\begin{figure}
\psfrag{c1}{$c_1$}
\psfrag{c2}{$c_2$}
    \centering
    \includegraphics[scale=0.7]{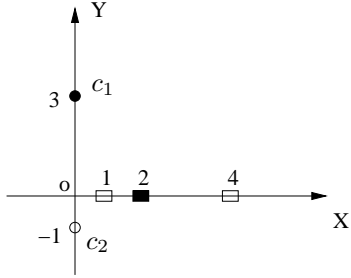}
    \caption{An example of optimal solution for 2-MST under the $L_1$ distance on the half-axis $(o,(\infty,0))$: in $T_1$, $c_1=(0,3)$
    takes the $2n$ points in the $2\varepsilon$-interval centered at $(2,0)$; in $T_2$, $c_2=(0,-1)$ takes the two groups of $n$ points in the
    $2\varepsilon$-intervals centered at $(1,0)$ and $(4,0)$. Following Lemma~\ref{lem3}, we could partition
    all the $n$ points near $(1,0)$ and the first $n$ points near $(2,0)$ to $T_1$ and the remaining points to $T_2$, without increasing the maximum
    weight of $T_1$ and $T_2$.}
  \label{fig:l1mst}
\end{figure}

If we denote a continuous segment of points of $P$ on the X-axis belonging to $T_1$ as $A$ and a segment of points of $P$ on the X-axis belonging to $T_2$ as $B$. The above lemma basically shows that in some optimal solution for 2-MST for this X+Y case, there is no
pattern like A-B-A on any of the half-axis in the X- and Y-axis.
Suppose there is an optimal solution with the A-B-A pattern:
making $c_1=(0,3)$ and $c_2=(0,-1)$ and three group of points (points are all within an interval of length $2\varepsilon$) around $(2,0)$ (with size $2n$), around
$(1,0)$ and $(4,0)$ (each with size $n$). One optimal solution is for $c_1$ to take the $2n$ points near $(2,0)$ and $c_2$ to take the remaining two groups of points (Fig.~\ref{fig:l1mst}). The optimal solution value is $5+\varepsilon$. But we could easily switch all the points near $(1,0)$ to $T_1$ and put the first half of $n$ points near $(2,0)$ to $T_1$. The weight of $T_2$ is unchanged and the weight of $T_1$ is decreased by $\varepsilon$.
We then have the following theorem.

\begin{theorem}
When all the points in $P$ and two sites $c_1$ and $c_2$ are on the X- and Y-axis, 2-MST under the $L_1$ metric can be solved in $O(n^4)$ time.
\end{theorem}

\begin{proof}
Following Lemma~\ref{lem3}, we can solve this problem in $O(n^4)$ time. We first sort the points of $P$ on the X-axis into $P_X$ and then we sort the points of $P$ on the Y-axis into $P_Y$. Then we enumerate all possible way to cut $P_X$ and $P_Y$ into at most 2 groups in each of the 4 half-axis. The total number is $O(n^4)$. 
Then, fixing each combination of cuts on the 4 half-axes, we check if a feasible solution exists, and if so, we compute the two MST's (including $c_1$ and $c_2$ respectively) in $O(1)$ time --- for each group we only need to compute its two extreme points when computing an MST. Consequently, we can compute the optimal solution of the 2-MST problem when all the points are on the X- and Y-axis in $O(n^4)$ time.
\end{proof}

\subsection{Points on the X- and Y-axis and under the Euclidean distance}

We now look at the X+Y case in this subsection by using the Euclidean distance. It turns out that the problem is much harder, as obviously not all the edges in an MST are along the X- and Y-axis.
In fact, different from the $L_1$ case, on any half-axis even the interleaving A-B-A scenario is possible for 2-MST in $L_2$ (Fig.~\ref{fig:2mst-worst}). However, we show that a pattern like A-B-A-B-A is not possible --- assuming $c_1$ and $c_2$ are not on the same half-axis. Based on that, we can give a polynomial time algorithm in $O(n^{13})$ time as well. First, we show a lemma regarding a property of an MST for points on the X- and Y-axis.

\begin{lemma}
When all the points in a set $Q$ are on the X- and Y-axis, in an MST of $Q$ under the $L_2$ metric, there are at most two consecutive segments of points of $Q$ on the X-axis (and respectively, Y-axis) not containing $c_i$.
\label{lem4}
\end{lemma}

\begin{proof}
In fact, we show a stronger statement: along any of the four half-axes not containing $c_i$, say $((0,0)$, $(+\infty,0))$, there is at most one segment of points in the MST $T$. WLOG, we refer to Fig.~\ref{fig:1mst}, where the MST connects two segments of
points through the edge $(a,d)$ and $(c,e)$. By triangle inequality, we could replace the edge $(c,e)$ with $(b,c)$. Then we would have a spanning tree with a smaller weight, as
$|ce|>|oc|>|bc|$. This contradicts the optimality of the assumed MST $T$.
\end{proof}
\begin{figure}
    \centering
    \includegraphics[scale=0.76]{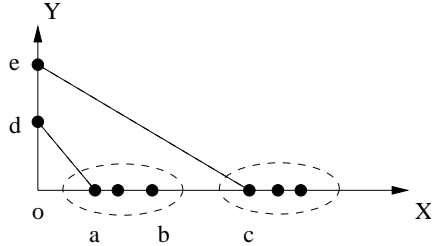}
    \caption{Illustration for the proof of Lemma~\ref{lem4}.}
    \label{fig:1mst}
\end{figure}

Note that the proof also implies that when computing the MST $T$, it all matters to identify the point closest to the origin $o$ in each of the half-axis, if $o$ is not in the input set $P$. We now explore more properties for 2-MST. 

\begin{lemma}
For the 2-MST problem under the $L_2$ metric, given each half-axis, say $((0,0)$, $(+\infty,0))$, except for the maximal segments connected with points on the Y-axis the optimal solution $T_1$ and $T_2$ must either partition the remaining points on the half-axis, possibly separated by ($c_i$, if any)
into two parts, or one of them takes all the points on it.
\label{lem5}
\end{lemma}

\begin{proof}
We focus on the half-axis $((0,0),(+\infty,0))$, and assume that the partition of points on this half-axis form five segments $[a,b]$, $[c,d]$, $[e,g]$, $[u,v]$ and $[w,z]$, where
$[a,b]$ and $[c,d]$ connect to some points/sites on the Y-axis, $[e,g]$ and $[w,z]$ belong to $T_1$ and $[u,v]$ belongs to $T_2$ (Fig.~\ref{fig:2mst} (I)). WLOG, assume that $c_1$ and $c_2$ are out of the interval $[c,z]$. In this case, similar to the proof of Lemma~\ref{lem3}, we show that we can decrease the number of segments of $T_1$ and $T_2$ without changing the connection $(a,h)$ and $(c,i)$
and without increasing the maximum weight of them. This can be done by partitioning the points in the segments/groups to the right of the last connection to the points in the Y-axis (i.e., segments $[e,g],[u,v]$ and $[w,z]$ to the right of point $c$ in Fig.~\ref{fig:2mst} (I)) into two parts; more precisely, partition these points into two parts according to the position of $c_1$ and $c_2$. In Fig.~\ref{fig:2mst} (II), when $c_1$ and $c_2$ are out of the interval $[c,z]$, then partition these points so that the leftmost
$|[u,v]|$ of them are merged with the segment $[c,d]$ for $T_1$ and the remaining ones are merged with $[w,z]$ for $T_2$. It is obvious that our goal is achieved.

Similar arguments obviously hold for the points between $c_1$ and $c_2$ (when $c_1$ and $c_2$ are on the same half-axis).
\end{proof}

\begin{figure}
\psfrag{n1}{$n_1$}
\psfrag{n2}{$n_2$}
\psfrag{n3}{$n_3$}
    \centering
    \includegraphics[scale=0.76]{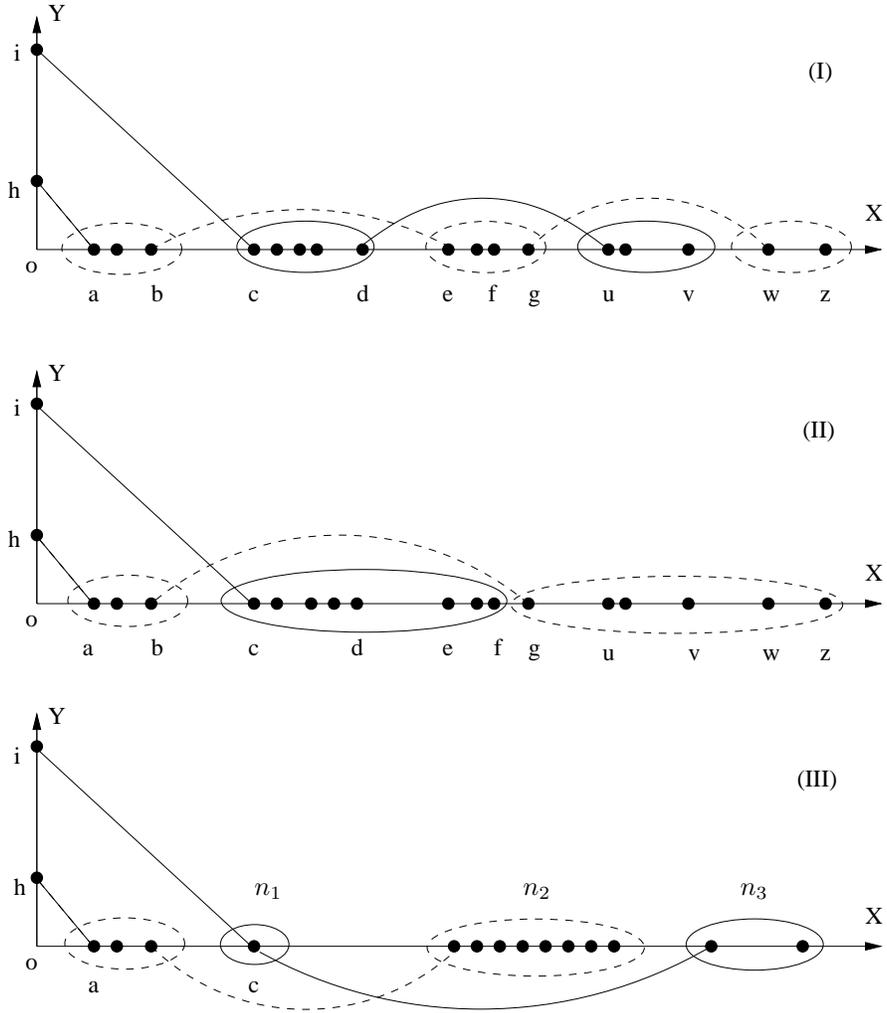}
    \caption{Illustration for the proof of Lemma~\ref{lem5}.}
    \label{fig:2mst}
\end{figure}

Fig.~\ref{fig:2mst} (III) shows that partition into two groups to the right of the segment containing $c$ could happen, as long as the number of points in the  rightmost three groups satisfy $n_2>n_1$ and $n_2>n_3$. This example cannot be further improved without changing the connection $(i,c)$ as in the example we set $n_1=1, n_2=8$ and $n_3=2$.
Note that the above lemma implies that, even excluding the segment bounded by $c_1$ and $c_2$ (when they are on the same half-axis), the pattern of A-B-A-B or B-A-B-A on any half-axis might still be possible, which enables us to design a polynomial-time algorithm. But we do not know yet if that pattern could really happen in real life. In Fig.~\ref{fig:2mst-worst}, we present an example where we do have the pattern A-B-A on an half-axis.

\begin{figure}
\psfrag{c1}{$c_1$}
\psfrag{c2}{$c_2$}
\psfrag{d1}{$9\sqrt{2}+1$}
    \centering
    \includegraphics[scale=0.7]{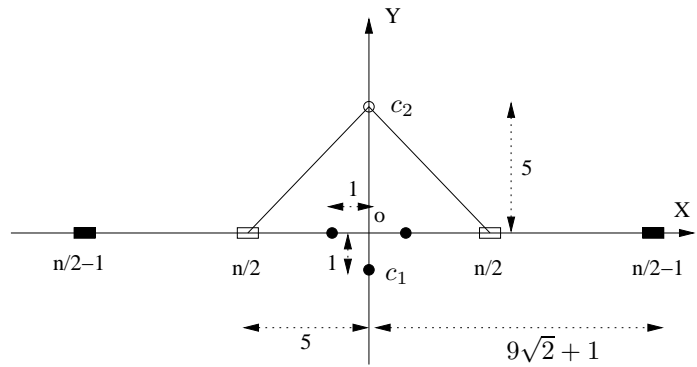}
    \caption{An example of optimal solution for 2-MST under the $L_2$ distance. In the example, $T_1$ includes black points composed of two blocks of $n/2-1$ points each (located in a small $2\varepsilon$-interval), plus two black points within distance 1 to the original $o$. They are all on the X-axis and together with $c_1=(0,-1)$ we form $T_1$, which has a weight of $10\sqrt{2}+2\varepsilon)$. $T_2$ is composed of two blocks of $n/2$ points on the X-axis (each within a $2\varepsilon$-interval located at a distance 5 from the origin), which are grouped with $c_2=(0,5)$ to form $T_2$. The weight of $T_2$ is also $10\sqrt{2}+2\varepsilon$. }
  \label{fig:2mst-worst}
\end{figure}

The algorithm for 2-MST for this X+Y case is then easy. First, ignore the case when $c_1$ and $c_2$ are on the same half-axis. We compute $T_1$ by at most 3-cutting the points and then selecting at most two segments along each of the 4 half-axes $((0,0),(+\infty,0))$, $((-\infty,0),(0,0))$, $((0,0),(0,+\infty))$, and $((0,0),(0,-\infty))$. This gives us $O((n^3)^4)=O(n^{12})$ number of partitions for $T_1$. Then if $c_1$ and $c_2$ are on the same half-axis, by Lemma~\ref{lem5}, we need one more cut to partition the points in between them. The total number of partitions for $T_1$ is $O(n^{13})$.
$T_2$ will then take the remaining segments. Hence, all pairs of $(T_1,T_2)$ can be enumerated in $O(n^{13})$ time. 
In an optimal solution such a set of at most 9 segments of points must exist, i,e., they cover exactly $n$ points and $c_1$. If we presort the points in the 4 half-axes, then this can be checked
in $O(1)$ time. Hence, $T_1$ can be computed in $O(1)$ time when its segments are given. Then, given each set of at most 9 (complementary) segments, we can compute the MST of the remaining points as $T_2$ in $O(1)$ time. This gives us the following theorem.

\begin{theorem}
When all the points in $P$ and two sites $c_1$ and $c_2$ are on the X- and Y-axis, 2-MST under the $L_2$ metric can be solved in $O(n^{13})$ time.
\end{theorem}

\section{Concluding Remarks}

In this paper, we focus the 2-MST problem which is a variation and generalization of the 2-squirrel problem we start with. While several results have been obtained, there are still many open questions. 
The first question is whether we could improve the approximation factor for 2-MST. 
Even with the current algorithm, we believe that the actual factor should be around 3. The second question is for the X+Y case of 2-MST under the Euclidean distance, we suspect that the $O(n^{13})$ upper bound is not
tight. There are possibly two ways to improve the bound: (1) if the pattern A-B-A-B on an half-axis (not containing $c_i$) can be shown to be impossible, then we only need at most two cuts
on each of them, leading to a running time of $O(n^9)$; (2) even if the pattern A-B-A-B on an half-axis (not containing $c_i$) is really possible, they might not appear in each half-axis
at the same time, then some improvement might still be possible.

\section*{Acknowledgments}

Part of this research was performed when the first and last author visited University of Hyogo in late 2022. We also thank Hiro Ito for some insightful comments.

\end{document}